\documentclass[conference]{IEEEtran}

\usepackage{amsthm}
\usepackage{amssymb}
\usepackage{float}
\usepackage[pdftex]{graphicx}
\usepackage{amsmath}
\usepackage{color}
\usepackage{slashbox}

\numberwithin{equation}{section}

\newtheorem{theorem}{Theorem}[section]

\newtheorem{definition}[theorem]{Definition}

\newtheorem{remark}[theorem]{Remark}

\begin{document}

\title{Convex Optimization Approach for Stable Decomposition of Stream of Pulses}
\author{\IEEEauthorblockN{Tamir Bendory}
\IEEEauthorblockA{Electrical engineering department \\
Technion - Israel Institute of Technology
}}
\maketitle

\begin{abstract}
This paper deals with the problem of estimating the delays and amplitudes of a weighted superposition of pulses, called
stream of pulses. This problem is motivated by a variety of applications, such as ultrasound and radar.
This paper shows that the recovery error of a tractable convex optimization problem is proportional to the noise level. Additionally, the estimated delays are clustered around the true delays. This holds provided that the pulse meets a few mild localization properties and that a separation condition holds. If the amplitudes are known to be positive, the separation is unnecessary. In this case, the recovery error is proportional to the noise level and depends on the maximal number of delays within a resolution cell.
\end{abstract}

\section{Introduction}
In many engineering and scientific problems, the acquired data is comprised of  a weighted super-position of pulses
(kernels). Typically, we aim to decompose the stream of pulses into its building blocks, frequently called \emph{atoms}. Ultrasound imaging \cite{tur2011innovation,wagner2012compressed} and Radar \cite{Bar-IlanSub-Nyquis} function as representative examples. In these applications, a pulse is transmitted and its echoes are reflected from different targets and recorded.

Mathematically, we consider the model 
\begin{equation}\label{eq:signal}
y[k]=\sum_{k\in\mathbb{Z}}c_mg_\sigma\left[k-k_m\right]+n[k],\quad c_m\in\mathbb{R},
\end{equation}
where $g_\sigma[k]:=g[k/\sigma]$ is a sampled version of the producing kernel $g(t)$ with a sampling spacing of $1/N$, namely $g[k]:=g(k/N)$, and $K:=\{k_m\}$. We assume that the error term is bounded $\Vert\mathbf{n}\Vert_1:=\sum_k\vert n[k]\vert \leq \delta$, with no additional statistical assumptions. Later on, we will also consider a positive stream of pulses, where the amplitudes are assumed to be positive $c_m>0$.  The acquired data (\ref{eq:signal}) can be presented as a spike deconvolution problem, i.e. 
\begin{equation*}
y[k]=\left(g_\sigma\ast x\right)[k]+n[k],
\end{equation*}
where $'*'$ denotes a discrete convolution, and 
\begin{equation} \label{eq:x}
x[k]=\sum_{k\in\mathbb{Z}}c_m\delta[k-k_m].
\end{equation}

The aim of this paper is to suggest a stable approach to decompose the stream of pulses into its atoms by solving a tractable convex optimization program. 

A well-known approach to decompose the signal into its atoms is by using parametric
methods, such as MUSIC, matrix pencil and ESPRIT \cite{stoica2005spectral,schmidt1986multiple,roy1989esprit,Matrix_pencil}. However, these
methods tend to be unstable in the presence of noise or model mismatch due to sensitivity
of polynomial root finding. 
An alternative way is to utilize compressed sensing and sparse representations theorems,
relying on the sparsity of the signal (e.g. \cite{donoho2006compressed,elad2010sparse}). However, these fields cannot explain the
success of $\ell_1$ minimization or greedy algorithms as the dictionaries have high coherence.

Inspired by recent advances in the theory of super-resolution \cite{candes2013towards,candes2013super,tang2013compressed,bendory2013exact,bendorySHalgorithm,bendory2013Legendre,de2012exact,azais2014spike,tang2015resolution}, we employed a convex optimization framework based on the existence of interpolating polynomials, frequently called \emph{the dual certificate}. In the next section, we elaborate on the convex optimization framework, and reveal the fundamental conditions, enabling stable decomposition of stream of pulses. Section \ref{sec:main} presents our main theorems. Particularity, we show that the solution of a convex optimization problem results in a stable and localized decomposition of stream of pulses under a separation condition if the pulse satisfies some mild localization properties. In the non-negative case, i.e. $c_m>0$, the separation is unnecessary and can be replaced by a weaker condition of Rayleigh regularity. 
We present all the results and the relevant definitions for a univariate stream of pulses, however we stress that similar results also exist for bivariate stream of pulses.
Ultimately, Section  \ref{sec:conclusions} concludes the work and suggests future extensions.

\section{Convex optimization approach for decomposition of stream of pulses}\label{sec:cvx}
In this paper we focus on a convex optimization approach for decomposing a stream of pulses in a noisy environment.
We use the Total-Variation (TV) norm as a sparse-promoting regularization. In essence, the TV norm is the generalization
of $\ell_1$ norm to the real line (for rigorous definition, see for instance \cite{rudin1986real}). For discrete measures of the form (\ref{eq:dis_measure}), we have $\Vert x\Vert_{TV}=\sum_m\vert c_m\vert$. 
The framework is based on a duality theorem, frequently called \emph{the dual certificate}, as follows \cite{bendorySOP}: 
\begin{theorem}
\label{th:dual}Let 
\begin{equation} \label{eq:dis_measure}
x(t)=\sum_{m}c_{m}\delta_{t_{m}}(t),\quad c_{m}\in\mathbb{R},\quad T:=\{t_{m}\}\subseteq\mathbb{R},
\end{equation}
 and let
$y(t)=\int_{\mathbb{R}}g(t-s)dx\left(  s\right)  $ for a $L$
times differentiable kernel $g(t)$. If for any set $\{v_{m}\}\in \{-1,1\}$, there exists a function of the form
\begin{equation}
q(t)=\int_{\mathbb{R}}\sum_{\ell=0}^{L}g^{(\ell)}(s-t)d\mu_{\ell}\left(
s\right)  ,\label{7}%
\end{equation}
for some measures \textup{$\left\{  \mu_{\ell}\left(  t\right)  \right\}
_{\ell=0}^{L}$}, satisfying
\begin{align*}
q(t_{m}) &  =v_{m}\,,\,\forall t_{m}\in T,\\
|q(t)| &  <1\,,\,\forall t\in\mathbb{R}\backslash T,%
\end{align*}
then $x$ is the unique real Borel measure solving
\begin{equation}
\min_{\tilde{x}\in\mathcal{M\left(  \mathbb{R}\right)  }}\Vert\tilde{x}%
\Vert_{TV}\quad\mbox{subject to}\quad y(t)=\int_{\mathbb{R}}g(t-s)d\tilde
{x}\left(  s\right)  .\label{eq:tv_min}%
\end{equation}
\end{theorem}
\begin{proof}
Let $\hat{x}$ be a solution of (\ref{eq:tv_min}), and define $\hat{x}=x+h$. The
difference measure $h$ can be decomposed relative to $|x|$ as
\[
h=h_{T}+h_{T^{C}},
\]
where $h_{T}$ is supported in $T$, and $h_{T^{C}}$ is supported in
$T^{C}$ (the complementary of $T$). 
If $h_T=0$, then also $h_{T^c}=h = 0$. Otherwise, $\Vert\hat{x}
\Vert_{TV}>\Vert{x}
\Vert_{TV}$ which is a contradiction. If $h_T\neq 0$, we 
perform a polar decomposition of
$h_{T}$ 
\[
h_{T}=|h_{T}|sgn(h_{T}),
\]
where $sgn(h_{T})$ is a function on $\mathbb{R}$ with values $\{-1,1\}$ (see
e.g. \cite{rudin1986real}). By assumption, for any $0\leq\ell\leq L$
\[
\int_{\mathbb{R}}g^{(\ell)}(t-s)d\hat{x}\left(  s\right)  =\int_{\mathbb{R}%
}g^{(\ell)}(t-s)dx\left(  s\right)  ,
\]
which in turn leads to $\int_{\mathbb{R}}g^{(\ell)}(t-s)dh\left(  s\right)
=0$. Then, for any $q$ of the form (\ref{7}) we get
\begin{align*}
\left\langle q,h\right\rangle  & =\int_{\mathbb{R}}q(t)dh\left(  t\right)  \\
& =\int_{\mathbb{R}}\left(  \int_{\mathbb{R}}\sum_{\ell=0}^{L}g^{(\ell
)}(s-t)d\mu_{\ell}\left(  s\right)  \right)  dh\left(  t\right)  \\
& =\int_{\mathbb{R}}\sum_{\ell=0}^{L}d\mu_{\ell}\left(  s\right)
\underbrace{\int_{\mathbb{R}}g^{(\ell)}(s-t)dh\left(  t\right)  }_{0}\\
& =0.
\end{align*}
By assumption, for the choice $v_{m}=$ $sgn(h_{T}\left(  t_{m}\right)  )$, there exists $q$ of the form (\ref{7}),
such that
\begin{align*}
q(t_{m}) &  =sgn(h_{T}(t_{m}))\,,\,\forall t_{m}\in T,\nonumber\\
|q(t)| &  <1\,,\,\forall t\in\mathbb{R}\backslash T.%
\end{align*}
Consequently,
\[
0=\left\langle q,h\right\rangle =\left\langle q,h_{T}\right\rangle
+\left\langle q,h_{T^{C}}\right\rangle =\Vert h_{T}\Vert_{TV}+\left\langle
q,h_{T^{C}}\right\rangle .
\]
If $h_{T^{C}}=0$, then $\Vert h_{T}\Vert_{TV}=0$, and $h=0$. Alternatively, if
$h_{T^{C}}\neq0$, from the second property of $q$,
\[
|\left\langle q,h_{T^{C}}\right\rangle |<\Vert h_{T^{C}}\Vert_{TV}.
\]
Thus,
\[
\Vert h_{T^{C}}\Vert_{TV}>\Vert h_{T}\Vert_{TV}.
\]
As a result, using the fact that $\hat{x}$ has minimal TV norm, we get
\[%
\begin{split}
\Vert x\Vert_{TV} &  \geq\Vert x+h\Vert_{TV}=\Vert x+h_{T}\Vert_{TV}+\Vert
h_{T^{C}}\Vert_{TV}\\
&  \geq\Vert x\Vert_{TV}-\Vert h_{T}\Vert_{TV}+\Vert h_{T^{C}}\Vert_{TV}>\Vert
x\Vert_{TV},
\end{split}
\]
which is a contradiction. Therefore, $h=0$, which implies that $x$ is the
unique solution of (\ref{eq:tv_min}).
\end{proof}

In practice, we cannot solve the infinite dimensional convex optimization problem (\ref{eq:tv_min}). Hence, we assume that the signal lies on a grid with spacing of $1/N$ which can be as fine as desired. In this case, the TV minimization (\ref{eq:tv_min}) reduces to standard $\ell_1$ minimization that can be solved by many existing solver. The solution of the discrete problem converges to
solution on the continuum (in the sense of measures) as the discretization becomes finer \cite{just_dis}. The behaviour of the solution in high SNR regime is analyzed in \cite{duval2013exact,duval2015sparse}.  

\section{Main results}\label{sec:main}

As aforementioned, the problem of decomposing a  stream of pulses can be reduced to the construction of an interpolating function, comprised of the kernel $g(t)$ and its derivatives.  The existence of such function relies on two interrelated pillars. First, the kernel should satisfy some localization properties, as follows: 
\begin{definition}
\label{def:admissible_kernel} A kernel $g$ is \emph{admissible} if it
has the following properties:

\begin{enumerate}
\item $g\in \mathcal{C}^3\left(\mathbb{R}\right)  $, is real and even.

\item \underline{Global property:} There exist constants $C_{\ell}>0, \ell=0,1,2,3$ such
that $\left\vert g^{\left(  \ell\right)  }\left(  t\right)  \right\vert
\leq\ C_{\ell} / \left( {1+t^{2}} \right)$ , where $ g^{\left(  \ell\right)  }\left(  t\right)$ denotes the $\ell^{th}$ 
derivative of  $g$.

\item \underline{Local property:} There exist constants $\varepsilon,\beta>0$ such that 
\begin{enumerate}
\item  $g(t)>0$ for all $\vert t\vert \leq \varepsilon$ and $g(t)  < g(\varepsilon)$ for all $\vert t\vert>\varepsilon$,
\item  $g^{\left(2\right)  }\left( t\right)  <-\beta$  for all $\vert t\vert \leq \varepsilon$.
\end{enumerate}

\end{enumerate}
\end{definition}
Two typical examples for admissible kernels are the Gaussian kernel $g(t)=e^{-\frac{t^2}{2}}$ and the Cauchy kernel $g(t)=\frac{1}{1+t^2}$ as presented in Table \ref{table1}.
\begin{table}
\begin{center}

    \begin{tabular}{| l | l | l |}
    \hline
   \backslashbox{}{Kernels} & Gaussian $:=e^{\frac{-t^2}{2}}$  & Cauchy $:=\frac{1}{1+t^2}$   \\ \hline
    $C_0$ & 1.22 & 1 
 \\ \hline
  $C_1$ & 1.59 & 1 \\ \hline
   $C_2$ & 2.04 & 2 \\\hline
      $C_3$ & 2.6 & 5.22 \\ \hline
      $g^{(2)}(0)$ & -1 & -2 \\ \hline      
            empirical $\nu$ & 1.1 & 0.45 \\ \hline      
    \end{tabular}
    \end{center}

    \caption{The table presents the numerical constants of the global property in Definition \ref{def:admissible_kernel} for the Gaussian and Cauchy kernels. Additionally, we evaluated by numerical experiments the minimal empirical value of $\nu$, the separation constant of Definition \ref{def:separation} for each kernel.  }
\label{table1}
\end{table}

 The second pillar is a kernel-dependent separation condition, as follows:

\begin{definition}
\label{def:separation} A set of points $K\subset\mathbb{Z}$ is said to satisfy
the minimal separation condition for a kernel-dependent $\nu>0$ and a given
$N,\sigma>0$ if%
\begin{equation*}
\min_{k_{i},k_{j}\in K,k_{i}\neq k_{j}}\left\vert k_{i}-k_{j}%
\right\vert \geq\nu\sigma N.
\end{equation*}
\end{definition}

In \cite{bendorySOP,SOP_US} we proved that if the kernel $g$ is admissible and the signal's support satisfies the kernel-dependent separation condition, then there exist constants $\{a_m\}$ and $\{b_m\}$ such that a function of the form 
\begin{equation*}
q(t)=\sum_m a_mg_\sigma(t-t_m)+g_\sigma^{(1)}(t-t_m),
\end{equation*}
satisfies the interpolation requirements of Theorem \ref{th:dual}. Hence, by Theorem \ref{th:dual} the $\ell_1$ minimization (assuming that the signal lies on the grid) recovers $x$ exactly from $y$. Additionally, the existence of the interpolating function is the key for proving the robustness and localization of the solution in a noisy environment. These results are summarized in the following theorem:

\begin{theorem}
\label{th:main} Consider the model \emph{(\ref{eq:signal}) } for an admissible kernel $g$.
Let us denote the solution of \begin{equation*}
min_{\tilde{x}}\left\Vert \tilde{x}\right\Vert _{1}\quad \mbox{subject to} \quad\left\Vert
y-g_{\sigma}\ast\tilde{x}\right\Vert _{1}\leq\delta,\label{6}%
\end{equation*}
as $\hat{x}=\sum_m\hat{c}_m\delta\left[k-\hat{k}_m\right]$ and $\hat{K}:=\left\{\hat{k}_m\right\}$.
 If $K$ satisfies
the separation condition of Definition \ref{def:separation} for  $N,\sigma >0$, then (for sufficiently large $\nu$)
\begin{equation} \label{eq:robust}
\left\Vert  \hat{x} - x \right\Vert _{1} \le \frac{16\gamma^2}{\beta }\delta,
\end{equation}
where $\gamma:=\max\left\{N\sigma,\varepsilon^{-1} \right\}$.
Additionally, if  $\varepsilon\geq\tilde{\varepsilon}:=\sqrt{\frac{g(0)}{C_{2}+\beta/4}}$ we have the following localization properties:
\begin{enumerate}
\item For any $k_{m}\in K$, if $c_{m}\geq2\delta D_{2}\left(1+\max\left\{ \frac{1}{D_{1}\varepsilon^{2}},\frac{4C_{2}}{\left(N\sigma\right)^{2}\beta}\right\} \right)$,
then there exists $\hat{k}_{m}\in\hat{K}$ such that 
\begin{equation*}
\begin{split}
&\left|k_{m}-\hat{k}_{m}\right|  \leq N\sigma \\ &\cdot\sqrt{\frac{2D_{2}\delta}{D_{1}\left(\left|c_{m}\right|-2\delta D_{2}\left(1+\max\left\{ \frac{1}{D_{1}\tilde{\varepsilon}^{2}},\frac{4C_{2}}{\left(N\sigma\right)^{2}\beta}\right\} \right)\right)}}.
\end{split}
\end{equation*}

\item 
\[
\sum_{\left\{ \hat{k}_{m}\in\hat{K}:\left|\hat{k}_{m}-k_{n}\right|>N\varepsilon\sigma,\forall k_{n}\in K\right\} }\left|\hat{c}_{m}\right|\leq\frac{2D_{2}}{D_{1}\varepsilon^{2}}\delta,
\]

\end{enumerate}
where 
\begin{eqnarray*}
D_{1}: & = & \frac{\beta}{4g\left(0\right)},\label{eq:D}\\
D_{2}: & = & \frac{3\nu^{2}\left(3\left|g^{(2)}(0)\right|\nu^{2}-\pi^{2}C_{2}\right)+\frac{16C_{1}^{2}\gamma^{2}\pi^{2}}{\beta}\left(1+\frac{\pi^{2}}{6\nu^{2}}\right)}{\left(3\left|g^{(2)}(0)\right|\nu^{2}-\pi^{2}C_{2}\right)\left(3g(0)\nu^{2}-2\pi^{2}C_{0}\right)}.\nonumber 
\end{eqnarray*}
\end{theorem}

\begin{remark}
A tighter estimation of (\ref{eq:robust}) can be found in \cite{bendorySOP}.
\end{remark}

In many applications, the underlying signal (\ref{eq:x}) is known to be non-negative, i.e. $c_m>0$. For instance, in single-molecule microscopy we measure the convolution of positive point sources with the microscope's point spread function \cite{klar2000fluorescence,betzig2006imaging,bronstein2009transient}.
It has become evident that in this case the separation is unnecessary and can be replaced be the weaker condition of Rayleigh regularity, defined as follows:

\begin{definition}
\label{def:rayleigh}We say that the set $\mathcal{P}\subset\left\{ k/N\right\} _{k\in\mathbb{Z}}\subset\mathbb{R}$
is Rayleigh-regular with parameters $(d,r)$ and write $\mathcal{P}\in\mathcal{R}^{idx}(d,r)$
if every interval $(a,b)\subset\mathbb{R}$ of length $\mu(a,b)=d$
contains no more that $r$ elements of $\mathcal{P}$:
\[
\left|\mathcal{P}\cap(a,b)\right|\leq r.
\]
\end{definition}

Equipped with Definition \ref{def:rayleigh}, we state the main theorem for the non-negative case. This result implies that the recovery error is proportional to the noise level $\delta$, and depends exponentially in the signal Rayleigh regularity $r$. 
\begin{theorem}
\label{th:main_positive} \cite{bendoryPositiveSOP} Consider the model (\ref{eq:signal}) with $c_m>0$ for an 
admissible kernel $g$ satisfying $g(t)\geq 0$. Then, there exists $\nu>0$ such that if $\mathcal{\mbox{supp(x)}}\in\mathcal{R}^{idx}\left(\nu\sigma,r\right)$
and $N\sigma>\left(\frac{1}{2}\right)^{\frac{1}{2r}+1}\sqrt{\frac{\beta}{g(0)}},$
the solution $\hat{x}$ of the convex problem 
\begin{equation*}
\min_{\tilde{x}}\quad\left\Vert \tilde{x}\right\Vert _{1}\quad\mbox{subject to}\quad\left\Vert y-g_\sigma\ast\tilde{x}\right\Vert _{1}\leq\delta,\thinspace\tilde{x}\geq0,\label{eq:cvx}
\end{equation*}
 satisfies (for sufficiently large $\nu$)
\begin{equation}
\left\Vert \hat{x}-x\right\Vert _{1}\leq\frac{2\left(2^{r}-1\right)}{C_{0}}\left(\frac{32C_{0}}{\beta}\right)^{r}\gamma^{2r} \delta,\label{eq:h_tight}
\end{equation}
where $\gamma:=\max\left\{ N\sigma,\varepsilon^{-1}\right\} $.
\end{theorem}
\begin{remark}
A tighter estimation of (\ref{eq:h_tight}) can be found in \cite{bendoryPositiveSOP}.
\end{remark}

\section{Conclusion} \label{sec:conclusions}
In this work we have shown that a standard convex optimization program can decompose a stream of pulses into its building blocks. In the general case (i.e. $c_m\in\mathbb{R}$), we have shown that the solution is robust in a noisy environment and that its support is clustered around the support of the sought signal. This holds provided that the convolution kernel $g$ is sufficiently localized and   a kernel-dependent separation condition holds. 

In the non-negative case, we have proven that the separation condition can be replaced be a weaker condition of Rayleigh regularity. The recovery error in this case is proportional to the noise level and depends on the signal's regularity. It is essentially important to derive the localization properties in this case as well.

The model presented in this work suits many practical applications where a signal is observed through a convolution kernel (typically, the point spread function of a sensing device). In previous work \cite{SOP_US}, we applied our algorithm for estimating the reflectors in \emph{in-vitro} ultrasound experiments. The experiments showed promising  results that corroborate our theoretical findings. It is of great interest to examine these theoretical results across more applications, particularly in the field of computational imaging.

\section*{Acknowledgment}

I would like to express my deep gratitude to Prof. Shai Dekel, Prof. Arie Feuer, Prof. Dan Adam and Avinoam Bar-Zion for the fruitful collaborations. 

\bibliography{bib}{}

\end{document}